\documentclass[conference]{IEEEtran}

\usepackage{cite} 
\usepackage{xspace}
\usepackage{filecontents}

\usepackage{pgf, tikz}

\usepackage{graphicx}

\usepackage[cmex10]{amsmath}

\usepackage{bm}

\usepackage{amssymb} 

\usepackage{array}

\usepackage{subfigure}
\usepackage{caption2} 

\usepackage{amsfonts}
\usepackage{amsthm}   

\newtheorem{proposition}{Proposition}

\newtheorem{remark}{Remark}
\usetikzlibrary{patterns,snakes}

\usepackage{enumerate}

\usepackage{extarrows}

\newcommand{\Ptx}{\mathcal{P}_{\mathrm{tx}}}
\newcommand{\Prx}{\mathcal{P}_{\mathrm{rx}}}

\newcommand{\Pth}{\mathcal{P}_{\mathrm{th}}}
\newcommand{\Pinfo}{\mathcal{P}_{\mathrm{tx},2}}
\newcommand{\Peh}{\mathcal{P}_{\mathrm{tx},1}}

\newcommand{\Pave}{\mathcal{P}_{\mathrm{ave}}}
\newcommand{\Ppeak}{\mathcal{P}_{\mathrm{peak}}}
\newcommand{\ser}{\mathrm{SER}}
\newcommand{\aveser}{\overline{\mathrm{SER}}}

\newcommand{\avessr}{\overline{\mathrm{SSR}}}

\newcommand{\rhoPS}{\rho_{\mathrm{PS}}}
\newcommand{\rhoTS}{\rho_{\mathrm{TS}}}
\newcommand{\mypower}{\bar{\mathcal{P}}}

\newcommand{\mypowerPS}{\bar{\mathcal{P}}_{\mathrm{PS}}}
\newcommand{\mypowerPSmax}{\bar{\mathcal{P}}_{\mathrm{PS},\max}}
\newcommand{\mypowerPSPSK}{\bar{\mathcal{P}}^{\mathrm{PSK}}_{\mathrm{PS}}}
\newcommand{\mypowerTS}{\bar{\mathcal{P}}_{\mathrm{TS}}}
\newcommand{\mypowerTSmax}{\bar{\mathcal{P}}_{\mathrm{TS},\max}}

\newcommand{\Plow}{P_{\mathrm{low}}}
\newcommand{\Phigh}{P_{\mathrm{high}}}
\newcommand{\qlow}{q_{\mathrm{low}}}
\newcommand{\qhigh}{q_{\mathrm{high}}}

\DeclareMathOperator*{\mymax}{\mathsf{Maximize}}
\DeclareMathOperator*{\myst}{\mathsf{Subject\ to}}

\begin{document}
\title{SWIPT with 
	Practical Modulation and RF Energy Harvesting Sensitivity		\vspace{-0.5cm}}

\author{\IEEEauthorblockN{Wanchun Liu$^*$, Xiangyun Zhou$^*$, Salman Durrani$^*$, Petar Popovski$^\dagger$}
	\IEEEauthorblockA{$^*$Research School of Engineering, The Australian National University, Canberra, Australia
		\\
		$^\dagger$Department of Electronic Systems, Aalborg University, Denmark
		\\
		Emails: \{wanchun.liu, xiangyun.zhou, salman.durrani\}@anu.edu.au,
		and petarp@es.aau.dk.
		\vspace{-0.5cm}
	}
}

\maketitle
\begin{abstract}
\let\thefootnote\relax\footnote{The work of W. Liu, X. Zhou and S. Durrani was supported by the Australian Research Council's Discovery Project Funding Scheme (project number DP140101133).
The work of P. Popovski has been in part supported by the European
Research Council (ERC Consolidator Grant Nr. 648382 WILLOW) within the
Horizon 2020 Program.
	}
In this paper, we investigate the performance of simultaneous wireless information and power transfer (SWIPT) in a point-to-point system, adopting practical $M$-ary modulation.
We take into account the fact that the receiver's radio-frequency (RF) energy harvesting circuit can only harvest energy when the received signal power is greater than a certain sensitivity level.
For both power-splitting (PS) and time-switching (TS) schemes,
we derive the energy harvesting performance as well as the information decoding performance for the Nakagami-$m$ fading channel.
We also analyze the performance tradeoff between energy harvesting and information decoding by studying an optimization problem, which maximizes the information decoding performance and satisfies a constraint on the minimum harvested energy.
Our analysis shows that 
(i) for the PS scheme, modulations with high peak-to-average power ratio achieve better energy harvesting performance,
(ii) for the TS scheme, it is desirable to concentrate the  power for wireless power transfer in order to minimize the non-harvested energy caused by the RF energy harvesting sensitivity level, and
(iii) channel fading is beneficial for energy harvesting in both PS and TS schemes.
\end{abstract}
\vspace{-0.1cm}
\section{Introduction}
Simultaneous wireless information and power transfer (SWIPT) has recently attracted significant attention~\cite{kaibin_zhou,WPT_survey}.
Practical SWIPT receivers for energy harvesting (EH) and information decoding (ID) have been proposed using power-splitting (PS) and time-switching (TS) schemes.
Current studies on SWIPT often consider an ideal information transmission model (i.e., Gaussian signaling) and investigate the tradeoff between the information capacity and harvested energy~\cite{ZhangHo,LiangLiuNew,MorsiJournal}.
In reality, SWIPT receivers are typically energy constrained and may be incapable of performing high-complexity  capacity-achieving coding/decoding scheme. 
Recently, SWIPT with practical coherent modulations was analyzed in \cite{XunZhou}.

Another commonly applied assumption in the SWIPT literature is that the average received signal power at the radio-frequency (RF) EH circuit is well above the RF-EH sensitivity level~\cite{ZhangHo,XunZhou,LiangLiuNew}.
Hence, these studies ignore the impact of the RF-EH sensitivity level.
In reality, practical state-of-the-art RF-EH circuits have power sensitivity requirement in the range of $-10$~dBm to $-30$~dBm~\cite{WPT_survey}.
Guaranteeing a much higher received signal power than the RF-EH sensitivity level often requires an extremely short communication range, which largely limits the application of SWIPT. Therefore, we consider a more general SWIPT system with $M$-ary modulation where the received signal power is not necessarily larger than the RF-EH sensitivity level. \emph{Since different constellation symbols may have different power levels, the amount of harvested energy may vary from symbol by symbol, and it is possible that some symbols can activate the RF-EH circuit but others cannot.}
Hence, it is important to accurately capture the effect of the RF-EH sensitivity level in analyzing the performance of SWIPT.

In this paper, we consider a transmitter-receiver pair adopting SWIPT with either PS or TS scheme.
In the PS scheme, the transmitter transfers modulated data signal to the receiver, then the receiver splits the received signal into two separate streams,
one to draw energy and one to acquire information, respectively.
In the TS scheme, for a given percentage of time, the transmitter transfers energy signal to the receiver, and the receiver draws energy from it. 
For the remaining portion of time, 
the transmitter transfers modulated data signal to the receiver, and the receiver acquires information from the signal.
Assuming a Nakagami-$m$ fading channel, we derive the average harvested power and the average number of successfully transmitted symbol per unit time (i.e., the symbol success rate), at the receiver.
We study the performance tradeoff between ID and EH using these metrics.
Our analysis provides the following interesting insights:
\begin{itemize}
	\item For the PS scheme, a modulation scheme with high peak-to-average power ratio (PAPR) leads to a better EH performance, e.g., $M$-PAM performs better than $M$-QAM, which performs better than $M$-PSK.
	This modulation performance order is very different from that of ID.
	\item For the TS scheme, we propose an optimal energy signal for the power transfer phase, which maximizes the available harvested energy. This is done by minimizing the duration of time for wireless power transfer since in this way one minimizes the amount of energy that is not harvested due to the impact of the RF-EH sensitivity~level.
	\item For both the PS and TS schemes, we show that channel fading is beneficial for EH when the RF-EH sensitivity level is considered in the analysis.
	This is in contrast~to previous studies, which ignored the sensitivity~level.
\end{itemize}

\vspace{-0.3cm}
\section{System Model}
We consider a SWIPT system consisting of a transmitter (Tx) and a receiver (Rx).
The receiver comprises an ID circuit and an RF-EH circuit~\cite{ZhangHo}.
Each node is equipped with a single omnidirectional antenna. 
The transmitter and receiver adopt block-wise operation with block time duration,~$T$.
We assume that the receiver is located in the far field, at a distance $d$ from the transmitter.
Thus, the channel link between the two nodes is composed of large scale path loss with exponent $\lambda$ and small-scale Nakagami-$m$ fading. 
Note that $m$ represents the fading parameter, which controls the severity of the fading.
The fading channel gain, $h$, is assumed to be constant within
one block time and independent and identically
distributed from one block to the next~\cite{ZhangHo,LiangLiuNew,Tianqing}. 
We assume instantaneous channel state information (CSI) is available only at the receiver.

We consider that the transmitter adopts a practical modulation scheme for information transmission (IT).
Let the signal constellation set be denoted by $\mathcal{X}$. The size of $\mathcal{X}$ is denoted by $M$ with $M=2^l$, and $l\geq 1$ being an integer.
The $i$th constellation point in $\mathcal{X}$ is denoted by $x_i$, $i=1,2,...,M$, with equal probability $p_i = 1/M$,\footnote{Consideration of modulation schemes with non-uniform probability distribution among all symbols is outside the scope of this work.} 
and the average power of signal set $\mathcal{X}$ is normalized to one, i.e., $\sum_{i=1}^{M} \vert x_i\vert^2 /M = 1$.
In this work, we consider three most commonly used coherent modulation schemes for IT: $M$-PSK, $M$-PAM, and $M$-QAM.\footnote{For $M$-QAM, we assume that $M=2^l$, and $l$ is an even integer.}

At the receiver,
during a symbol period $T_s$,
assuming the receive power at the RF-EH circuit is $\Prx$,
the amount of harvested energy can be represented as~\cite{Tianqing}
\vspace{-5pt}
\begin{equation} \label{RFEH_model}
\mathcal{E} = \eta T_s \left(\Prx-\Pth\right)^+,
\vspace{-7pt}
\end{equation}
where $0 \leq \eta \leq 1$ is the RF-EH efficiency, and $(z)^+ = \max \left\lbrace z,0 \right\rbrace$. We assume that the harvested energy at the receiver is stored in an ideal battery~\cite{ZhangHo,Tianqing,RFEH}.
Note that in \eqref{RFEH_model},
according to the existing studies~\cite{Tianqing,RFEH,NewThreshold},
the RF-EH circuit can only harvest energy when its 
receive signal power, $\Prx$, is greater than the RF-EH sensitivity level,
$\Pth$, and the harvested energy is proportional to $\Prx- \Pth$.

\subsection{PS and TS Schemes}
The operation of SWIPT using PS or TS scheme is described as follows:
\begin{itemize}
	\item \textit{PS}: In each time block, the transmitter sends modulated data signal with average transmit power $\Ptx$. The receiver splits the received signal with a
	PS ratio, $\rhoPS$, for separate EH and ID.
	\item \textit{TS}: 
	Each time block is divided into a power transfer (PT) phase and an IT phase.
	First, the transmitter transmits an energy signal with average transmit power $\Peh$ during the first $\rhoTS T$ seconds, i.e., PT phase. The receiver harvests energy from this received signal.
	Then, the transmitter transmits modulated signal with average transmit power $\Pinfo$ during the remaining $(1-\rhoTS)T$ seconds, i.e., IT phase, and the receiver decodes the received information.
%
\end{itemize}

\subsection{Transmit Power Constraints}
We consider both average and peak power constraints at the transmitter for SWIPT~\cite{ZhangHo,LiangLiuNew}, denoted by $\Pave$ and $\Ppeak$, respectively.
From~\cite{PpeakHigh} and references therein, 
the peak-to-average ratio of a practical RF circuit can be as large as $13$~dB.
In this paper, we assume that $\Ppeak/\Pave \geq 3$, i.e., $4.77$~dB.
\subsubsection{PS}
First, the average transmission power $\Ptx$ should satisfy the average power constraint, i.e., 
$\Ptx \leq \Pave$.

Second, different transmitted symbols have different power. Thus, the highest transmit power of the symbols in $\mathcal{X}$ should also be no larger than $\Ppeak$. 
From~\cite{Book_fading}, the PAPR of $M$-PSK/PAM/QAM 
are given by
\begin{equation} \label{PAPR_modulation}
\varphi_{\mathrm{PSK}} =1,\varphi_{\mathrm{PAM}} =3 \frac{M-1}{M+1}, \varphi_{\mathrm{QAM}}=3 \frac{\sqrt{M}-1}{\sqrt{M}+1},
\end{equation}
respectively.
Therefore, peak power constraint should be satisfied as 
$
\varphi_{j} \Ptx \leq \Ppeak,\ j=\mathrm{PSK, PAM, QAM}.
$

From~\eqref{PAPR_modulation}, we see that the PAPR is always less than~$3$.
Thus, for the considered modulation schemes, the peak power constraint is automatically satisfied, as long as the average power requirement is met, i.e., $\Ptx \leq \Pave$.

\subsubsection{TS}
First, the transmit power should satisfy average power constraint as $\rhoTS\Peh +(1-\rhoTS)  \Pinfo \leq \Pave$.

Second, for the PT phase of the TS scheme, the peak power of the energy signal should be no larger than  $\Ppeak$, and a constraint for the average power of the energy signal, $\Peh\leq \Ppeak$, is satisfied naturaly.
For the IT phase of the TS scheme, as discussed above, peak power constraint should be satisfied~as
\begin{equation} \label{varphi}
\varphi_{j} \Pinfo \leq \Ppeak,\ j=\mathrm{PSK, PAM, QAM}.
\end{equation}

\subsection{Metrics}
For EH, we use the average harvested power $\mypower$ as a performance metric~\cite{MorsiJournal,LiangLiuNew}.
For ID, we adopt the average success symbol rate, $\avessr$, which is related to the average symbol error rate, $\aveser$, as the performance metric. The $\avessr$ measures the average number of successfully transmitted symbols per unit time.
Using $\mypower$ and $\avessr$, we investigate the performance tradeoff between EH and ID.
The exact mathematical definitions of $\mypower$ and $\avessr$ are given in Sections III and IV for the PS and TS schemes, respectively.
%
%

\section{Power Splitting}
Following \cite{ZhangHo,AliOld}, after PS, the received RF signal at the RF-EH circuit is 
\begin{equation} \label{PS_EH_signal}
y(t)= h(t) \sqrt{\frac{\rhoPS \Ptx  }{d^{\lambda}}}\ x(t) + \tilde{n}(t),
\end{equation}
where 
$h(t)$ represents the channel fading gain as described in Sec. II,
$x(t)$ is the modulated information signal sent from the transmitter, and $\tilde{n}(t)$ is the circuit noise.

For ID, the RF signal is down converted to baseband. The received signal and the signal noise ratio (SNR) at the ID circuit of the receiver are given by\footnote{In \eqref{ID_signal}, for notational simplification, we have represented the baseband signals, $y[k]$, $x[k]$ and $n[k]$ as $y$, $x$ and $n$, respectively, where $k$ denotes symbol index.}
\begin{equation}\label{ID_signal}
y= h \sqrt{\frac{(1-\rhoPS) \Ptx }{d^{\lambda}} }\ x +n,
\end{equation}
and
\vspace{-5pt}
\begin{equation}
\gamma(\nu) = \frac{(1-\rhoPS) \Ptx \nu}{d^{\lambda} \sigma^2},
\end{equation}
respectively,
where $n$ is the AWGN at the ID circuit with power $\sigma^2$,
and $\nu = \vert h \vert^2$ is the fading power gain of the channel.
Since we consider Nakagami-$m$ fading, the probability density function (pdf) of $\nu$ is given by
\begin{equation} \label{Nakagami_fading}
	f_{\nu}(v) = \frac{v^{m-1} m^m}{\Gamma(m)} \exp\left(-mv\right),\ m=1,2,3....
\end{equation}

In the following, we analyze the performance of EH and ID with $M$-PSK, $M$-PAM and $M$-QAM schemes.

\vspace{-5pt}
\subsection{RF Energy Harvesting with the PS Scheme}

Using \eqref{RFEH_model} and \eqref{PS_EH_signal}, and assuming that the amount of energy harvested from the circuit noise $\tilde{n}$ is negligible~\cite{ZhangHo,XunZhou},
the harvested energy under fading power gain $v$ during one symbol time $T_s$ is given by
%
\vspace{-5pt}
\begin{equation} \label{PS_energy_first}
\mathcal{E}(x_i,v) = \eta T_s \left( \frac{\rhoPS \mathcal{P}_i v }{d^{\lambda}} - \Pth\right)^+,
\vspace{-5pt}
\end{equation}
where $\mathcal{P}_i$ is the transmit power for the $i$th constellation point, and $\sum_{i=1}^{M} \mathcal{P}_i/M = \Ptx$. From~\cite{Book_fading}, $\mathcal{P}_i$ can be obtained as
%
\vspace{-5pt}
\begin{equation} \label{P_i}
\begin{aligned}
&\mathcal{P}_i = \left\lbrace 
\begin{aligned}
&\Ptx, &&\hspace{-1.5cm}\text{ for $M$-PSK}\\
&\frac{3\Ptx}{M^2-1}\left(2 \left\lceil \left\vert i-\frac{M+1}{2} \right\vert \right\rceil  - 1 \right)^2, &&\hspace{-1.5cm}\text{ for $M$-PAM}\\
&\frac{3\Ptx}{2(M-1)} \left((2 \left\lceil \left\vert \left\lceil \frac{i}{\sqrt{M}} \right\rceil-\frac{\sqrt{M}+1}{2} \right\vert \right\rceil - 1)^2 \right.\\
&\left.+(2 \left\lceil \left\vert \left(i\!\mod \sqrt{M} \right)-\frac{\sqrt{M}+1}{2} \right\vert   \right\rceil - 1)^2\right)\!\!, \\
& &&\hspace{-1.5cm}\text{for $M$-QAM}
\end{aligned}
\right.\\
\end{aligned}
\end{equation}
where $\lceil \cdot \rceil$ is the ceiling operator, and the operator $z\!\!\mod{Z}$ has the definition as follows: if $z$ is an integer multiple of $Z$, 
$z\mod{Z}$ is $Z$, while in other cases, it is equal to $z$~modulo~$Z$.

From \eqref{PS_energy_first},
the average harvested power can be calculated as
\begin{equation} \label{ave_EH_power_first}
\begin{aligned}
\mypowerPS 
& = \frac{1}{T_s}\sum_{i=1}^{M} p_i \int_{0}^{\infty} \mathcal{E}(x_i,v) f_{\nu}(v) \mathrm{d}v  \\
&\hspace{-0.7cm} =\eta \sum_{i=1}^{M} p_i
\int_{\bar{v}_i}^{\infty}
\left( \frac{\rhoPS \mathcal{P}_i v }{d^{\lambda}} - \Pth\right) f_{\nu}(v) \mathrm{d} v,\\
\vspace{-5pt}
& \hspace{-0.7cm}\!=\!\eta \! \sum_{i=1}^{M} p_i
\!\left(\!
\frac{\rhoPS \mathcal{P}_i }{d^{\lambda}}\!\! \left(1\!\!-\!\!\int_{v=0}^{\bar{v}_i} \!\!v f_{\nu}(v) \mathrm{d} v\!\right) \!
\!-\!\Pth\! \left(\!1\!- \!\!\int_{v=0}^{\bar{v}_i}\! f_{\nu}(v) \mathrm{d} v\!\right)\!
\right)\!,
\end{aligned}
\vspace{-5pt}
\end{equation}
where $p_i$ is the transmit probability of symbol $x_i$, i.e., $1/M$, and
$
\bar{v}_i = {\Pth d^{\lambda}}/{\rhoPS \mathcal{P}_i}.
$
By taking \eqref{Nakagami_fading}, \eqref{P_i} and $p_i=1/M$ into \eqref{ave_EH_power_first}, after some simplification, we get
\begin{equation} \label{PS_ave_EH_power}
\begin{aligned}
\vspace{-5pt}
\mypowerPS 
&= \frac{\eta}{M} \sum_{i=1}^{M}\frac{\rhoPS \mathcal{P}_i}{d^{\lambda}} \exp\left(-m \Pth \frac{d^{\lambda}}{\rhoPS \mathcal{P}_i}\right) \times\\
\vspace{-5pt}
&\hspace{0.5cm}\left(1 + \sum\limits_{k=0}^{m-1} \left(\Pth \frac{d^{\lambda}}{\rhoPS \mathcal{P}_i} \right)^{k+1} \frac{m^k (m-k-1)}{\left(k+1\right)!}
\right).
\end{aligned}
\end{equation}

\underline{Special case}: For $m=1$, i.e., Rayleigh fading channel, 
we have
\vspace{-5pt}
\begin{equation} \label{PS_rayleigh}
\mypowerPS = \frac{\eta}{M} \sum_{i=1}^{M} \frac{\rhoPS \mathcal{P}_i}{d^{\lambda}} \exp(- \frac{\Pth d^{\lambda}}{\rhoPS \mathcal{P}_i}).
\end{equation}
For $M$-PSK, \eqref{PS_rayleigh} further simplifies to
\begin{equation} \label{PS_PSK}
\mypowerPSPSK = \eta \frac{\rhoPS \Ptx}{d^{\lambda}} \exp\left(- \frac{\Pth d^{\lambda}}{\rhoPS \Ptx}\right).
\end{equation}

\begin{remark}
For the PS scheme, a modulation scheme with higher PAPR, or a fading channel with larger variation in its power gain, increases the average harvested power.	
\end{remark}
The above observations in Remark 1 will be verified in Sec.~V.
They can be explained using the analysis as follows:

Due to the RF-EH sensitivity level, $\Pth$, \eqref{PS_energy_first} is a convex function w.r.t. symbol transmit power, $\mathcal{P}_i$, and channel power fading gain $v$.
Thus, based on Jensen's inequality, 
taking \eqref{PS_energy_first} into the first line of \eqref{ave_EH_power_first}, we have
\begin{equation} \label{my_jensen}
\begin{aligned}
\mypowerPS 
&\geq 
\eta \left(\frac{\rhoPS}{d^\lambda} \sum_{i=1}^{M}p_i \mathcal{P}_i \int_{0}^{\infty} v f_{\nu} (v) \mathrm{d}v -\Pth \right)^+.
\vspace{-15pt}
\end{aligned}
\end{equation}
We see that the equality holds in \eqref{my_jensen}, i.e., average harvested power $\mypowerPS$ is minimized, iff the variance of both the symbol power in the signal set and the channel fading power gain equal to zero, i.e., $\mathcal{P}_i = \Ptx$ for all $i$, and $v$ is always equal to one (non-fading channel). 

For practical modulations, $M$-PSK results in the worst EH performance since it has constant symbol power (lowest PAPR and zero variance of symbol power).
In addition, our numerical results suggest that $M$-PAM (which has highest PAPR) performs better than $M$-QAM, which performs better than $M$-PSK. This will be verified in Sec. V.

For the Nakagami-$m$ fading, we can easily prove using \eqref{PS_ave_EH_power} that $\mypowerPS$ decreases with increasing $m$.
As $m$ increases, the fading variance decreases.
Hence, a channel with larger fading variance (smaller $m$) increases the average harvested power.

\subsection{Information Decoding with the PS Scheme}
The average symbol success rate, $\avessr$, can be calculated~as
\begin{equation} \label{PS_throughput_final}
\avessr = \int_{v=0}^{\infty} \left(1- {\ser(v)}\right) f_{\nu}(v) \mathrm{d}v
= 1-\aveser.
\end{equation}
where
\begin{equation} \label{PS_AVE_SER}
\aveser = \int_{v=0}^{\infty} \mathrm{SER}(v) f_{\nu}(v) \mathrm{d}v,
\end{equation}
and $\ser(v)$ is given by \cite{Book_fading} as\footnote{For simplicity, we adopt an accurate and widely used approximate closed-form expression for $M$-PSK.}
\begin{equation} \label{PS_SER}
	\mathrm{SER}(v)= \left\lbrace 
	\begin{aligned}
		&2 \ Q\left(\sqrt{2 g_{\mathrm{PSK}} \gamma(v)}\right), \hspace{1.8cm} \text{for $M$-PSK}\\
		&2 {(M-1)}/{M} \ Q\left(\!\sqrt{2g_{\mathrm{PAM}} \gamma(v)}\right), \text{for $M$-PAM}\\
		&4\left(1-{1}/{\sqrt{M}}\right)\ Q\left(\sqrt{2g_{\mathrm{QAM}} \gamma(v)}\right)\\
		\vspace*{-4pt}
		&-4\left(1-{1}/{\sqrt{M}}\right)^2\ Q^2\left(\sqrt{2g_{\mathrm{QAM}} \gamma(v)}\right), \\
		\vspace{-7pt}
		&\hspace{4.9cm}\text{for $M$-QAM}\\
	\end{aligned}
	\right.
\end{equation}
where $Q(\cdot)$ is the Q-function, $M$ is the modulation order, $g_{\mathrm{PSK}} = \sin^2(\pi/M)$, 
$g_{\mathrm{PAM}} = \frac{3}{M^2-1}$
and $g_{\mathrm{QAM}} = \frac{3}{2(M-1)}$.

\begin{remark}
	\textnormal{Given $\rhoPS$,  $M$ and $v$, it is known that in the high SNR regime, $M$-QAM outperforms $M$-PSK, which outperforms $M$-PAM on SER~\cite{Book_fading}. Thus, from \eqref{PS_throughput_final}, the same order holds for average success symbol rate.}
\end{remark}

\subsection{Optimal Transmission Strategy for the PS Scheme}
It is interesting to investigate the problem that given a certain required average harvested power level, $\mypower_0$, what is the achievable average symbol success rate ($\avessr$). 
For the PS scheme, there are two design parameters: the PS ratio, $\rhoPS$, and the transmit power, $\Ptx$.

We consider the following optimization problem:
\begin{equation} \label{PS_Opt_1}
\begin{aligned}
\mathbf{P1}:\ \ &\mymax_{\rhoPS,\Ptx} \ \ \ &&\avessr = 1-\aveser \\
&\myst\ &&\mypowerPS \geq \mypower_0,  \\
& &&0 \leq \rhoPS  \leq 1,
 \ 0 \leq \Ptx  \leq \Pave.
\end{aligned}
\end{equation}
where 
$\mypowerPS$ is defined in \eqref{PS_ave_EH_power},
and the required average harvested power, $\mypower_0$, is no higher than the maximum range of average harvested power, $\mypowerPSmax$, which is obtained by taking $\rhoPS = 1$ into \eqref{PS_ave_EH_power}.

The optimization problem can be solved as follows. 
From \eqref{PS_throughput_final} and \eqref{PS_ave_EH_power}, it is easy to see that both $\mypowerPS$ and $\avessr$ are monotonically increasing functions w.r.t. $\Ptx$. 
Meanwhile, $\mypowerPS$ and $\avessr$ increase and decrease with $\rhoPS$, respectively.
Thus, in $\mathbf{P1}$, we let $\Ptx = \Pave$ and $\mypowerPS = \mypower_0$.
In order to obtain the optimal $\rhoPS$, $\rhoPS^{\star}$, which yields $\mypowerPS = \mypower_0$, a method of bisection for linear search can be adopted.
Taking $\rhoPS^{\star}$ into \eqref{PS_throughput_final}, the maximum achievable average symbol success rate, $\avessr^{\star}\!$, is obtained.

\underline{Special case}: For $M$-PSK scheme with Rayleigh fading channel,
we can obtain a closed-form expression for $\rhoPS^{\star}$.
By taking $\Ptx = \Pave$ into \eqref{PS_PSK} and letting \eqref{PS_PSK}~$=\mypower_0$, and solving we get
\begin{equation}
\rhoPS^{\star} =\frac{\Pth d^{\lambda}}{\Pave W_0\left(\frac{\eta \Pth}{ \mypower_0}\right)},
\end{equation}
where $W_0(\cdot)$ is the principle branch of the Lambert W function.

\section{Time Switching}
In the PT phase, i.e., $\rhoTS$ portion of a time block, 
the transmitter transmits a pre-designed energy signal with average power $\Peh$.
Thus, the received RF signal at the RF-EH circuit is given by
\begin{equation}
y(t)= h(t) \sqrt{ \frac{1}{d^{\lambda}}}\ \omega(t) + \tilde{n}(t),
\end{equation}
where $\omega(t)$ is the pre-designed energy signal, which will be investigated in the following subsection.

In the IT phase, i.e., $1-\rhoTS$ portion of a time block, the baseband received signal and the SNR at the ID circuit of the receiver are given by
\begin{equation}
y= h \sqrt{\frac{ \Pinfo }{d^{\lambda}} }\ x +n,
\end{equation}
and
\begin{equation}\label{TS_snr}
\gamma(\nu) = \frac{\Pinfo \nu}{d^{\lambda} \sigma^2},
\end{equation}
respectively.

In the following, we analyze the performance of EH and~ID.

\subsection{RF Energy Harvesting with the TS Scheme}
If the RF-EH sensitivity level is zero, any energy signal design will achieve the same harvested energy. On the other hand, with a non-zero RF-EH sensitivity level, different energy signals perform differently.
Here, we propose an optimal energy signal which maximizes the harvested energy.
\begin{proposition}
	\textnormal{For the PT phase, given the average transmit power, $\Peh$, and the RF-EH sensitivity level, $\Pth$, an optimal energy signal which achieves maximum harvested energy, is the signal with power $\Ppeak$ for the first ${\Peh}/{\Ppeak}$ portion of the phase, 
	and with power zero for the rest of the phase.}
\end{proposition}
\begin{proof}
Although the proposed optimal energy signal seems simple,
to the best of our knowledge,
prior studies have not provided an optimal energy signal design under the RF-EH model in \eqref{RFEH_model}.
We provide the formal proof in Appendix A.
\end{proof}
\begin{remark}
\textnormal{In order to maximize the harvested energy with the consideration of RF-EH sensitivity level, 
	the available power/energy is concentrated into the shortest possible duration of time, i.e., transmit with $\Ppeak$.
	In this way, the amount of energy that is wasted during the harvesting process is minimized.
}
\end{remark}

Similar to Sec. III-A, the average harvested power for the TS scheme is given by
\begin{equation} \label{TS_ave_EH_power}
\begin{aligned}
\mypowerTS 
&= \eta \rhoTS \frac{\Peh}{\Ppeak} \int_{0}^{\infty} \left( \frac{\Ppeak v}{d^{\lambda}} -\Pth \right)^+ f_{\nu}(v)\mathrm{d}v\\
&= \eta \rhoTS \frac{\Peh}{\Ppeak} \int_{\bar{v}_{\mathrm{peak}}}^{\infty} \left( \frac{\Ppeak v}{d^{\lambda}} -\Pth \right) f_{\nu}(v) \mathrm{d}v\\
& = \eta \rhoTS \frac{\Peh }{d^{\lambda}}\  \Psi,
\end{aligned}
\end{equation}
where $\bar{v}_{\mathrm{peak}} = {\Pth d^{\lambda}}/{\Ppeak}$, and 
\begin{equation} \label{my_Psi}
\begin{aligned}
&\Psi=\exp\left(-m\frac{\Pth d^\lambda}{\Ppeak}\right) \times \\
& \hspace{1.2cm}\left(1 + \sum\limits_{k=0}^{m-1} \left(\Pth \frac{d^{\lambda}}{\Ppeak} \right)^{k+1} \frac{m^k (m-k-1)}{\left(k+1\right)!}\right).
\end{aligned}
\end{equation}
Because of the average power constraint, $\rhoTS\Peh +\left(1-\rhoTS\right) \Pinfo \leq \Pave$,
taking $\rhoTS\Peh = \Pave$ into \eqref{TS_ave_EH_power},
the maximum range of average harvested power is given by
\begin{equation} \label{TS_max_power}
\begin{aligned}
\mypowerTSmax
= \eta  \frac{\Pave }{d^{\lambda}}\  \Psi.
\end{aligned}
\end{equation}

Similar to the PS scheme, from \eqref{TS_ave_EH_power}, we can easily prove that $\mypowerTS$ decreases with increasing $m$. Thus, fading channel with larger variation in its power gain will also increase the EH performance of the TS scheme, as it does for the PS scheme.

\subsection{Information Decoding with the TS Scheme}
Compared to the PS scheme, only $1-\rhoTS$ portion of a block time is used for IT in the TS scheme. 
Since the $\avessr$ in the PS scheme measures the average number of successfully transmitted symbols per unit time,
we consider a normalized $\avessr$ for the TS scheme in order to measure the same quantity, which is given by
\begin{equation} \label{TS_throughput}
\begin{aligned}
	\avessr &= (1-\rhoTS) \int_{v=0}^{\infty} \left(1-\ser(v) \right) f_{\nu}(v) \mathrm{d}v\\
	&= (1-\rhoTS)  \left(1- \aveser\right),
\end{aligned}
\end{equation}
where $\ser(v)$ and $\aveser$ are obtained by taking \eqref{TS_snr} into \eqref{PS_SER} and \eqref{PS_AVE_SER} respectively.

\subsection{Optimal Transmission Strategy for the TS Scheme}
We investigate the problem that given a certain required minimum average harvested power level, $\mypower_0$, what is the maximum achievable average symbol success rate.
For the TS scheme, there are three design parameters: the TS ratio, $\rhoTS$, the power for PT, $\Peh$, and the power for IT, $\Pinfo$.
We consider the following optimization problem:
\begin{equation} \label{TS_Opt_1}
\begin{aligned}
\mathbf{P2}:\ \ &\mymax_{\rhoTS,\Peh,\Pinfo} \ \ \ &&\avessr = (1-\rhoTS)  \left(1-\aveser\right)\\
&\myst\ &&\mypowerTS \geq \mypower_0,  \\
& &&\rhoTS\Peh +\left(1-\rhoTS\right) \Pinfo \leq \Pave,\\
& && 0 \leq \Peh, \varphi_{j}\Pinfo \leq \Ppeak, \\
& && 0 \leq \rhoTS \leq 1,\\
\end{aligned}
\end{equation}
where $\mypower_0$ is no higher than the maximum range of average harvested power $\mypowerTSmax$ defined in~\eqref{TS_max_power}.
$\mypowerTS$ and $\varphi_j$ are defined in \eqref{TS_ave_EH_power} and \eqref{varphi}, respectively.

\begin{proposition}
\textnormal{For the TS scheme, given $\mypower_0$, $\Pave$, $\Ppeak$ and $\Pth$, the optimal parameters are given by
\begin{equation} \label{TS_opti_parameter}
\begin{aligned}
&\rhoTS^{\star} = \frac{\mypower_0 d^{\lambda}}{\eta \Ppeak \Psi},\ 
\Peh^{\star} = \Ppeak,\\
&\Pinfo^{\star} = \frac{1}{1-\rhoTS^{\star}} \left(\Pave - \frac{\mypower_0 d^{\lambda}}{\eta\ \Psi }\right),\\
\end{aligned}
\end{equation}
where $\Psi$ is defined in \eqref{my_Psi}.
}
\end{proposition}
\begin{proof}
See Appendix B.
\end{proof}
\noindent Taking \eqref{TS_opti_parameter} into \eqref{TS_throughput}, optimal achievable average symbol success rate is obtained.
\begin{remark}
\textnormal{Recall that Proposition 1 suggests to transfer energy at a peak power and then turn the energy signal off during the remaining time if any.
Contrary to this, after performing a joint optimization as done in Proposition 2, the resulting energy signal for PT is transmitted at constant power $\Ppeak$ which occupies the entire PT~phase.
		}
\end{remark}

\section{Numerical Results}
\begin{figure*}[t]
	\renewcommand{\captionfont}{\small} \renewcommand{\captionlabelfont}{\small}
	\minipage{0.47\textwidth}
	\includegraphics[width=\linewidth]{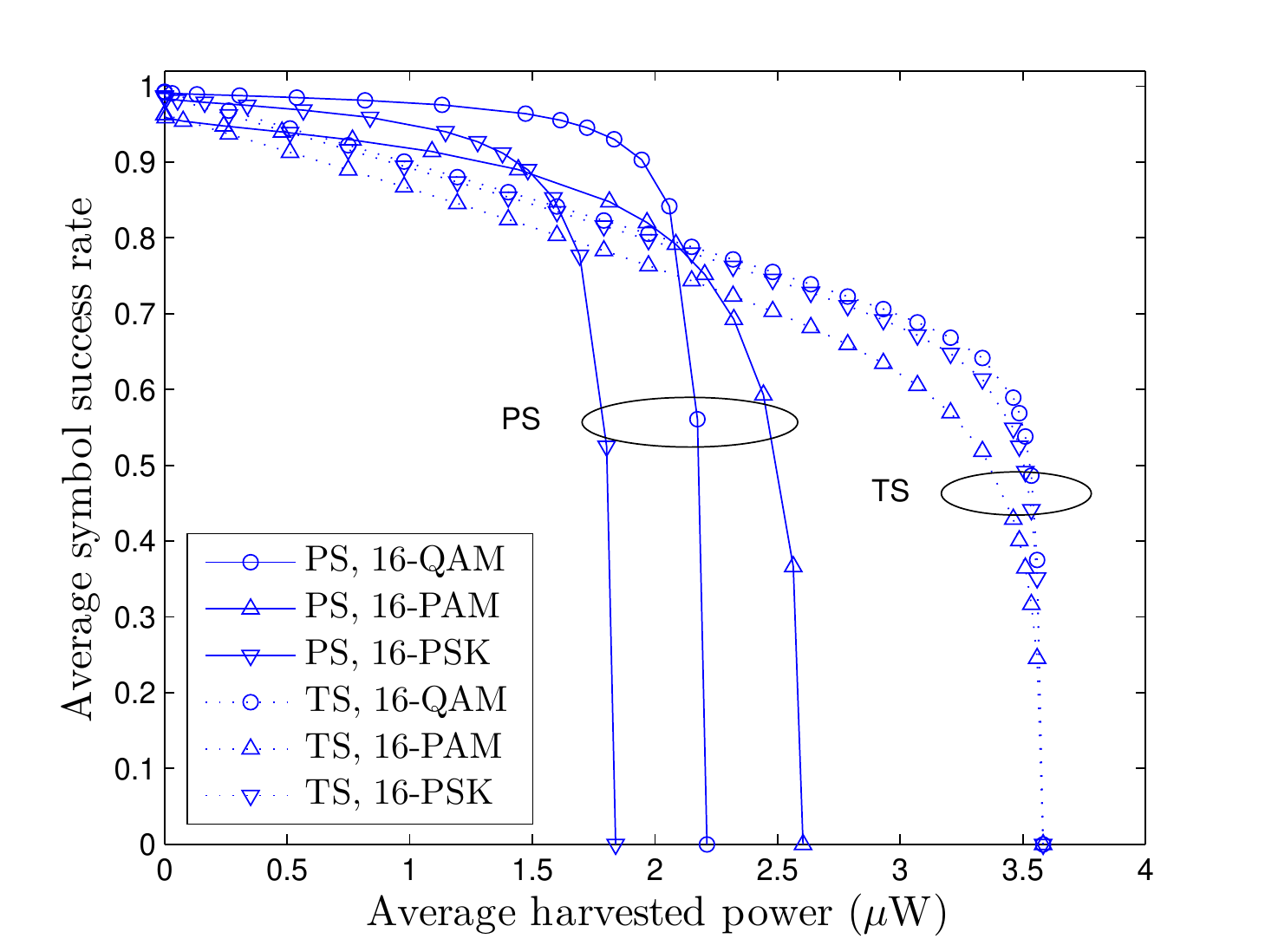}
	\caption{Tradeoff between the average symbol success rate and the average harvested power with different modulation schemes, $M$=16.}\label{fig:M-16}
	\endminipage\hfill
	\minipage{0.47\textwidth}
	\hspace{-0.5cm}
	\includegraphics[width=\linewidth]{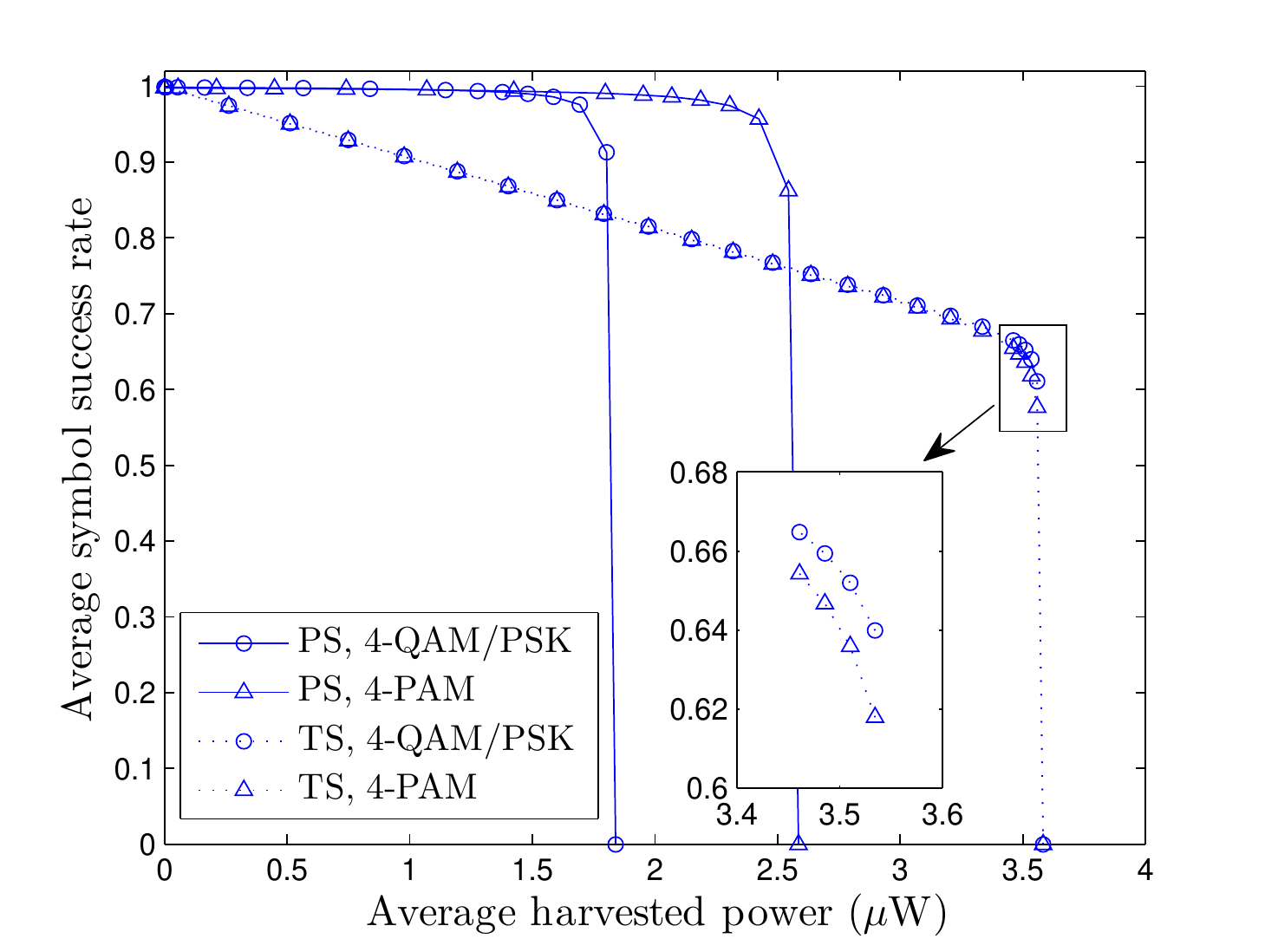}
	\caption{Tradeoff between the average symbol success rate and the average harvested power with different modulation schemes, $M$=4.}\label{fig:M-4}	
	\endminipage\\
	\hrulefill
	\vspace*{-10pt}
\end{figure*}
In this section, we illustrate the tradeoffs between the  average symbol success rate and the average harvested power for the PS and TS schemes, which are obtained by using the optimal transmission strategies given in Sec.~III-C and Sec.~IV-C, respectively. 

We set the average power constraint at the transmitter as $\Pave = 10$~mW, 
peak power constraint as $\Ppeak = 3\Pave$~\cite{Ppeak3},
distance between the transmitter and the receiver as $d =10$~m, 
path loss exponent as $\lambda=3$~\cite{Tianqing}. 
We set the noise power at the receiver ID circuit as $\sigma^2 = -50$~dBm.
The receiver RF-EH conversion efficiency is $0.5$~\cite{WPT_survey}.
Unless otherwise stated,
we set
the modulation order at the transmitter as $M=16$,
the RF-EH sensitivity level as $\Pth=-20$~dBm~\cite{WPT_survey}, 
the Nakagami-$m$ fading parameter as $m=1$, i.e., Rayleigh fading channel.
Under these setting, it can be proved that the average receive power at the receiver is no less than the RF-EH sensitivity level~\cite{ZhangHo}.

\subsection{Effect of Modulation Schemes}
Figs. \ref{fig:M-16} and \ref{fig:M-4} plot the tradeoffs between the average symbol success rate and the average harvested power using PS and TS schemes and different modulation schemes, with $M=16$ and $4$, respectively.



For the PS scheme, from Figs. \ref{fig:M-16} and \ref{fig:M-4}, we see that the different modulation schemes result in different performance tradeoff between ID and EH.
When the required average harvested power is high,
$M$-PAM performs better than $M$-QAM, and $M$-QAM performs better than $M$-PSK.\footnote{When $M=4$, $M$-PSK is equivalent with $M$-QAM.} This order matches their order on PAPR, which is in line with Remark~1.
When the required average harvested power is low, $M$-QAM performs better than $M$-PSK, and $M$-PSK performs better than $M$-PAM. This order matches their order on average symbol success rate in high SNR regime, which is in accordance with Remark~2.
Thus, if we aim to harvest energy, a modulation scheme with higher PAPR performs better than the one with lower PAPR, while if we focus on information decoding, the modulation scheme with the highest PAPR performs the worst.

For the TS scheme, from Figs. \ref{fig:M-16} and \ref{fig:M-4}, we see that the effect of different modulation schemes is smaller than that with the PS scheme. This is because only IT phase is affected by modulation.
Thus, in general (high SNR regime), $M$-QAM performs better than $M$-PSK, and $M$-PSK performs better than $M$-PAM on average symbol success rate, as stated in Remark~2.

Also we see that given a modulation scheme, there is a crossover between the tradeoff curves of the PS and TS schemes. When the average harvested power requirement is high, the TS scheme always results in a better ID performance than the PS scheme. This is due to the optimized energy signal for the TS scheme in Proposition 1, which makes the EH more efficient under the presence of RF-EH sensitivity level, while for the PS scheme, the signal for EH cannot be optimized.
Note that this crossover has not been identified in the rate-energy tradeoff curves studied in~\cite{ZhangHo}.

\subsection{Effect of RF-EH Sensitivity}
Fig. \ref{fig:Pth} 
plots the tradeoffs between the average symbol success rate and the average harvested power under different RF-EH sensitivity level, using $16$-QAM modulation scheme.
We see that given a target average harvested power, the average symbol success rate varies significantly between $\Pth =0$ and practical RF-EH sensitively level, i.e., $\Pth =-20$~dBm. 
When $\Pth =0$, there is no crossover between the tradeoff curves of the PS and TS schemes, and the range of average harvested power is the same between the two schemes. 
This is similar to the results shown in~\cite{ZhangHo}, which ignored the effect of $\Pth$.
If a practical RF-EH sensitivity level is set, e.g., $\Pth = -20$~dBm, the crossover is clearly observed, and we see that different SWIPT schemes result in different ranges of average harvested power. Thus, under a practical RF-EH sensitivity level, the performance tradeoff between ID and EH is very different from that obtained under ideal assumption.
This highlights the importance of considering the RF-EH sensitivity level in this work.

\subsection{Effect of Fading Channel}
Fig. \ref{fig:Fading} plots the tradeoffs between the average symbol success rate and the average harvested power using $16$-QAM modulation scheme, for different Nakagami-$m$ fading channel parameters.
The figure confirms that different values of $m$ do not affect the general trends of the curves,
as identified in Figs. 1-3 for $m=1$. 
We see that if the required average harvested power is high, the channel with larger non-line-of-sight (NLOS) component,
e.g., $m=1$,
leads to better ID performance than the channel with smaller NLOS component, e.g., $m \rightarrow \infty$.
However, if the required average harvested power is low, the channel with larger NLOS component,  leads to worse ID performance than the channel with smaller NLOS component.
This implies that for SWIPT, fading is beneficial for EH but detrimental for ID.
This also verifies our observation in Remark~1.

\section{Conclusion}
In this paper, we studied SWIPT between a transmitter and a receiver, using either PS or TS scheme,
taking practical modulation and receiver RF-EH sensitivity level into account.
For both the PS and TS schemes, we designed the optimal system parameters,
which satisfy transmit average and peak power constraints and maximize the ID performance with a constraint of minimum harvested energy. 
Our analysis showed that 
for the PS scheme,
a modulation scheme with higher PAPR achieves better EH performance.
For the TS scheme, we proposed an optimal energy signal for the PT phase, which maximizes the available harvested energy. 
In addition, 
channel fading is beneficial for EH when considering realistic values of the RF-EH sensitivity level.
The results highlight the importance of  accurately characterizing the impact of RF-EH sensitivity level in SWIPT systems. 
Future work can consider the impact of non-constant RF-EH efficiency for the performance of SWIPT.

\section*{Appendix A: Proof of Proposition 1}
Given channel fading power gain $v$, average power for PT $\Peh$, time for PT $\tau$,
we assume that an energy signal is a $N$-level power signal, with power levels in ascending sort as, $\tilde{P}_i$, $i=1,2,..., N$. $\tilde{P}_i$ holds for $q_i$ percentage of the time duration, $\tau$, and we have
$
\sum_{i=1}^{N} q_i \tilde{P}_i = \Peh.
$

Assuming that $\tilde{P}_{n-1} \leq \Pth$ and $\tilde{P}_{n} > \Pth$, we propose the following optimization problem to find the optimal energy signal which maximize the harvested energy during $\tau$ as:
\begin{equation}
\begin{aligned}
\mathbf{PA}:\ \ &\mymax_{\Plow,\Phigh,\qlow,\qhigh} && \qhigh \left(\Phigh - \Pth\right)\\
&\hspace{15pt}\myst\ && \hspace{-5pt} \qlow \Plow +\qhigh \Phigh = \Peh,  \\
& &&\hspace{-5pt} 0 \leq \Plow \leq \Pth \leq \Phigh \leq \Ppeak,\\
& &&\hspace{-5pt} \qlow +\qhigh = 1, \qlow,\qhigh \geq 0,
\end{aligned}
\end{equation}
where $\qlow = \sum_{i=1}^{n-1} q_i$, $\qhigh = \sum_{i=n}^{N} q_i$, 
$\Plow = \sum_{i=1}^{n-1} q_i \tilde{P}_i/\qlow$
and $\Phigh = \sum_{i=n}^{N} q_i \tilde{P}_i/\qhigh$.

Because the target function of $\mathbf{PA}$ is a monotonically increasing function w.r.t. both $\qhigh$ and $\Phigh$, we let $\qlow \Plow =0$ in the first constraint of $\mathbf{PA}$, i.e., $\qhigh \Phigh = \Peh$. Then, $\mathbf{PA}$ can be easily solved as
$
\Phigh^\star = \Ppeak,  \qhigh^{\star} = \Peh/\Ppeak,
\Plow^\star = 0, \text{ and } \qlow^{\star} = 1-\Peh/\Ppeak.
$
Therefore, an optimized energy signal is proposed in Proposition~1.
\begin{figure}[t]
	\renewcommand{\captionlabeldelim}{ }
	\renewcommand{\captionfont}{\small} \renewcommand{\captionlabelfont}{\small}
	\centering 
	\includegraphics[scale=0.55]{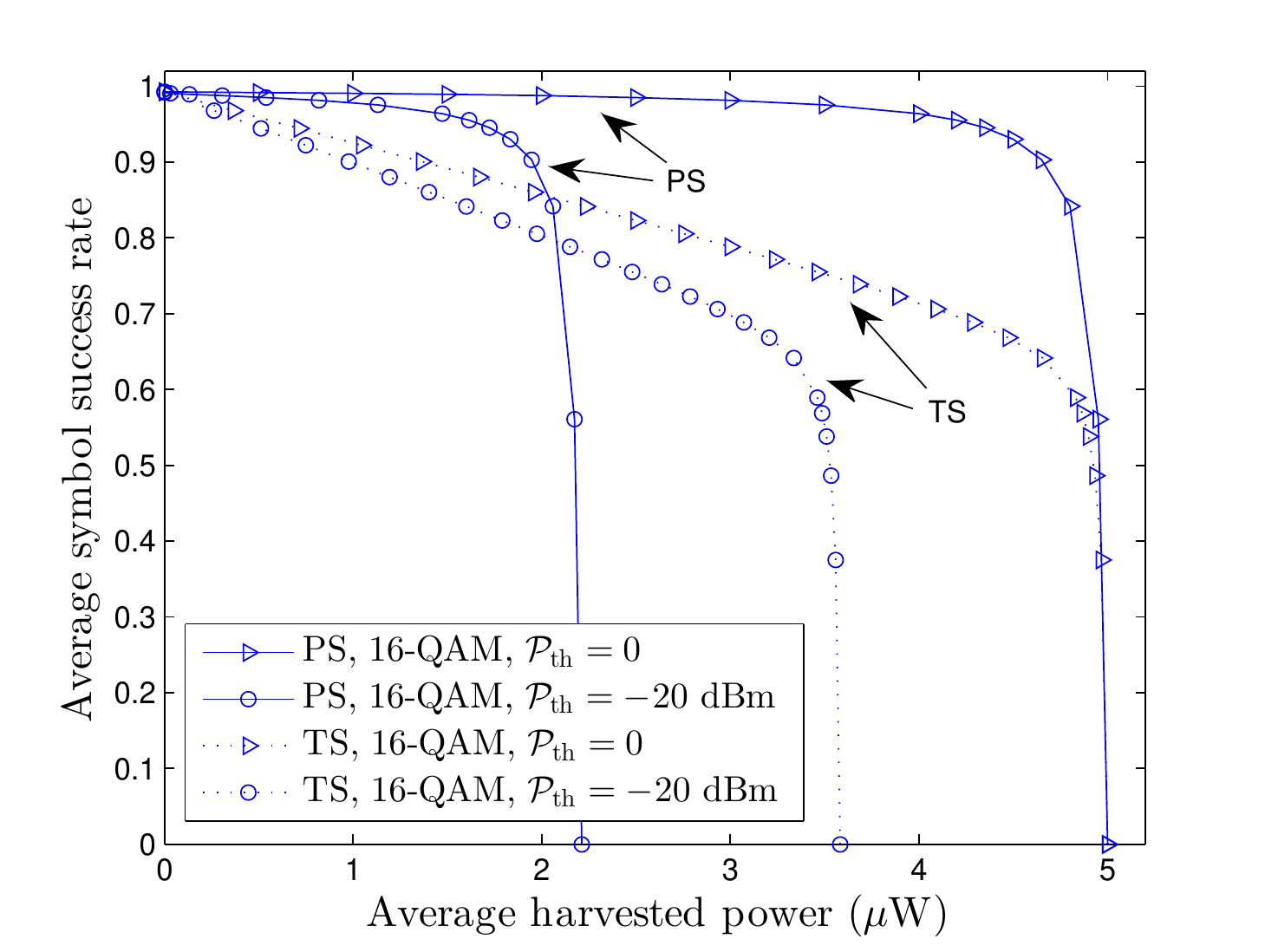}	
	\caption{Tradeoff between the average symbol success rate and the average harvested power with different RF-EH sensitivity level.}	
	\label{fig:Pth}
	\vspace*{-14pt}
\end{figure}

\section*{Appendix B: Proof of Proposition 2}
In order to solve $\mathbf{P2}$, we first relax the peak power constraint for $\Pinfo$, i.e., 
replace $\varphi_i \Pinfo \leq \Ppeak$ with $\Pinfo \leq \Ppeak$. 

From \eqref{PS_SER}, $1-\ser(v)$ is a monotonically increasing function w.r.t. $\Pinfo$. 
Thus, from \eqref{TS_throughput}, the target function of $\mathbf{P2}$, $\avessr$, is a monotonically increasing and decreasing function w.r.t. $\Pinfo$ and $\rhoTS$, respectively.
Therefore, the first and the second constraints of $\mathbf{P2}$ are active constraints, i.e., the equalities hold, which are given by (after simplification)
\begin{align}
\label{cons_1}
&{\rhoTS \Peh } = \frac{\mypower_0 d^{\lambda}}{\eta \Psi},\\
\label{cons_2}
&(1-\rhoTS)  = \frac{1}{\Pinfo} \left(\Pave -\frac{\mypower_0 d^{\lambda}}{\eta \Psi}\right),
\end{align}
respectively. Taking \eqref{cons_2} into \eqref{TS_throughput}, we have
\begin{equation} \label{App_ssr}
\avessr = \int_{0}^{\infty} \frac{1}{\Pinfo} \left(1-\ser(v)\right) f_{\nu}(v) \mathrm{d} v.
\end{equation}
Given $v$, taking \eqref{PS_SER} into \eqref{App_ssr} and using the property that $Q'(x) = -\frac{1}{\sqrt{2\pi}} \exp\left(-x^2/2\right)$, we verify that $\frac{1}{\Pinfo} \left(1-\ser(v)\right)$ is a monotonically decreasing function w.r.t. $\Pinfo$. Thus, $\avessr$ in \eqref{App_ssr} monotonically decreases with $\Pinfo$.
From \eqref{cons_2}, in order to minimize $\Pinfo$, we need to find the minimum $\rhoTS$, which can be obtained by letting $\Peh = \Ppeak$ in \eqref{cons_1}.
Thus, the optimal $\rhoTS$ can be derived, and then the optimal $\Pinfo$ is obtained from \eqref{cons_2}, which are shown in Proposition~2.

From the above solution, we obtain the optimal $\Pinfo$ which is no larger than $\Pave$. Based on our explanation for peak power constraint in Sec. II-B, the optimal $\Pinfo$ satisfies the constraint $\varphi_i \Pinfo \leq \Ppeak$ naturely. Thus,  the calculated optimal $\Pinfo$ is exactly the global optimal $\Pinfo$.

\begin{figure}[t]
	\renewcommand{\captionlabeldelim}{ }
	\renewcommand{\captionfont}{\small} \renewcommand{\captionlabelfont}{\small}
	\centering 
	\includegraphics[scale=0.55]{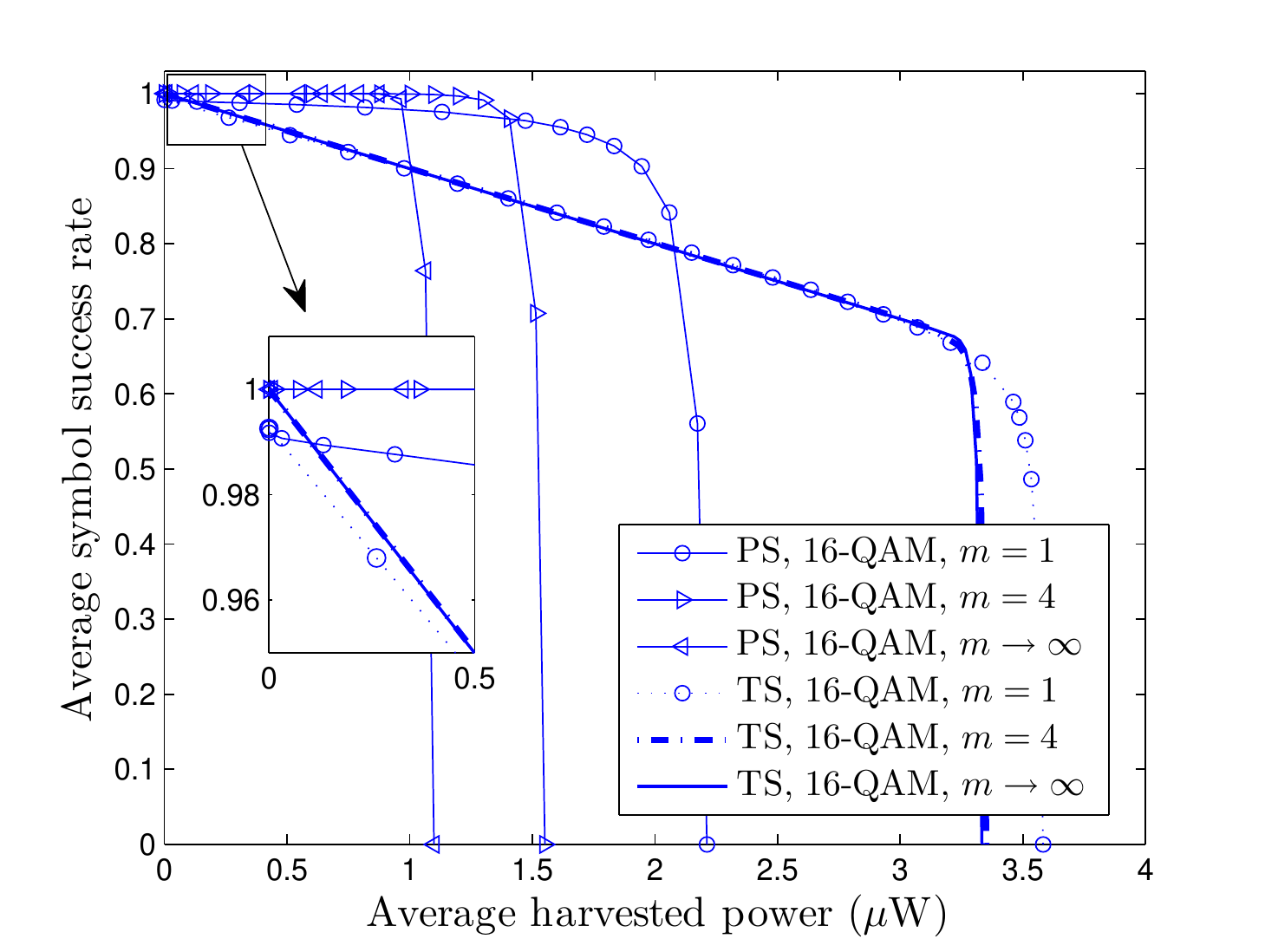}	
	\caption{Tradeoff between the average symbol success rate and the average harvested power with different Nakagami-$m$ fading parameters.}	
	\label{fig:Fading}
	\vspace*{-14pt}
\end{figure}

\ifCLASSOPTIONcaptionsoff
\fi


\begin{thebibliography}{10}
\providecommand{\url}[1]{#1}
\csname url@samestyle\endcsname
\providecommand{\newblock}{\relax}
\providecommand{\bibinfo}[2]{#2}
\providecommand{\BIBentrySTDinterwordspacing}{\spaceskip=0pt\relax}
\providecommand{\BIBentryALTinterwordstretchfactor}{4}
\providecommand{\BIBentryALTinterwordspacing}{\spaceskip=\fontdimen2\font plus
\BIBentryALTinterwordstretchfactor\fontdimen3\font minus
  \fontdimen4\font\relax}
\providecommand{\BIBforeignlanguage}[2]{{%
\expandafter\ifx\csname l@#1\endcsname\relax
\typeout{** WARNING: IEEEtran.bst: No hyphenation pattern has been}%
\typeout{** loaded for the language `#1'. Using the pattern for}%
\typeout{** the default language instead.}%
\else
\language=\csname l@#1\endcsname
\fi
#2}}
\providecommand{\BIBdecl}{\relax}
\BIBdecl

\bibitem{kaibin_zhou}
K.~Huang and X.~Zhou, ``Cutting last wires for mobile communication by
  microwave power transfer,'' \emph{IEEE Commun. Mag.}, vol.~53, no.~6, pp.
  86--93, Jun. 2015.

\bibitem{WPT_survey}
L.~Xiao, P.~Wang, D.~Niyato, D.~Kim, and Z.~Han, ``Wireless networks with {RF}
  energy harvesting: A contemporary survey,'' \emph{IEEE Commun. Surveys
  Tuts.}, vol.~17, no.~2, pp. 757--789, Second Quarter 2015.

\bibitem{ZhangHo}
R.~Zhang and C.~K. Ho, ``{MIMO} broadcasting for simultaneous wireless
  information and power transfer,'' \emph{{IEEE} Trans. Wireless Commun.},
  vol.~12, no.~5, pp. 1989--2001, May 2013.

\bibitem{LiangLiuNew}
L.~Liu, R.~Zhang, and K.-C. Chua, ``Wireless information and power transfer: A
  dynamic power splitting approach,'' \emph{{IEEE} Trans. Commun.}, vol.~61,
  no.~9, pp. 3990--4001, Sep. 2013.

\bibitem{MorsiJournal}
R.~Morsi, D.~Michalopoulos, and R.~Schober, ``Multi-user scheduling schemes for
  simultaneous wireless information and power transfer over fading channels,''
  \emph{IEEE Trans. Wireless Commun.}, vol.~14, no.~4, pp. 1967--1982, Apr.
  2015.

\bibitem{XunZhou}
X.~Zhou, R.~Zhang, and C.~K. Ho, ``Wireless information and power transfer:
  Architecture design and rate-energy tradeoff,'' \emph{{IEEE} Trans. Commun.},
  vol.~61, no.~11, pp. 4754--4767, Nov. 2013.

\bibitem{Tianqing}
T.~Wu and H.-C. Yang, ``On the performance of overlaid wireless sensor
  transmission with {RF} energy harvesting,'' \emph{{IEEE} J. Sel. Areas
  Commun.}, vol.~33, no.~8, pp. 1693--1705, Aug. 2015.

\bibitem{RFEH}
U.~Baroudi, A.~Qureshi, V.~Talla, S.~Gollakota, S.~Mekid, and A.~Bouhraoua,
  ``Radio frequency energy harvesting characterization: an experimental
  study,'' in \emph{Proc. IEEE International Conference on Trust, Security and
  Privacy in Computing and Communications (TrustCom)}, Feb. 2012, pp.
  1976--1981.

\bibitem{NewThreshold}
A.~Boaventura and N.~Carvalho, ``Maximizing {DC} power in energy harvesting
  circuits using multisine excitation,'' in \emph{Proc. IEEE MTT-S Int. Microw.
  Symp. Dig.}, Jun. 2011, pp. 1--4.

\bibitem{PpeakHigh}
M.~Khoshnevisan and J.~Laneman, ``Power allocation in wireless systems subject
  to long-term and short-term power constraints,'' in \emph{Proc. IEEE ICC},
  Jun. 2011, pp. 1--5.

\bibitem{Book_fading}
M.~K. Simon and M.~S. Alouini, \emph{Digital Communication over Fading
  Channels: A Unified Approach to Performance Analysis}.\hskip 1em plus 0.5em
  minus 0.4em\relax Wiley, 2000.

\bibitem{AliOld}
A.~Nasir, X.~Zhou, S.~Durrani, and R.~Kennedy, ``Relaying protocols for
  wireless energy harvesting and information processing,'' \emph{{IEEE} Trans.
  Wireless Commun.}, vol.~12, no.~7, pp. 3622--3636, Jul. 2013.

\bibitem{Ppeak3}
S.~de~la Kethulle~de Ryhove and G.~Oien, ``Rate-optimal power adaptation in
  average and peak power constrained fading channels,'' in \emph{Proc. IEEE
  WCNC}, Mar. 2007, pp. 1773--1778.

\end{thebibliography}

\end{document}